\documentclass[a4paper,USenglish,10pt]{article}
\usepackage[utf8]{inputenc}

\usepackage{amsthm,amssymb,amsfonts,amsmath}
\usepackage{amssymb,amsfonts,amsmath}
\newtheorem{theorem}{Theorem}
\newtheorem{lemma}[theorem]{Lemma}
\newtheorem{corollary}[theorem]{Corollary}

\newtheorem{observation}[theorem]{Observation}
\theoremstyle{definition}

\usepackage[hidelinks]{hyperref}
\usepackage[blocks]{authblk}

\usepackage{graphicx}
\graphicspath{{figures/}}
\usepackage{wrapfig}

\newcommand{\fig}[1]{\figurename~\ref{#1}}

\makeatletter
\g@addto@macro\bfseries{\boldmath}
\makeatother

\advance\hoffset-7mm
\usepackage{thm-restate}

\newcommand{\td}{\ensuremath{\textsf{td}}}
\newcommand{\gtd}{\ensuremath{\textsf{gtd}}}

\usepackage{url}

\newcommand{\CH}{\ensuremath{\textsf{CH}}}

\newcommand{\ceil}[1]{\ensuremath{\left \lceil #1 \right \rceil}}
\newcommand{\floor}[1]{\ensuremath{\left \lfloor #1 \right \rfloor}}

\newcommand{\R}{\ensuremath{\mathbb{R}}}

\makeatletter
\g@addto@macro\bfseries{\boldmath}
\makeatother

\usepackage{authblk}

\date{}

\title{Extending the centerpoint theorem to multiple points}

\begin{document}
\author[1]{Alexander Pilz\thanks{Supported by a Schr\"odinger fellowship of the Austrian Science Fund (FWF): J-3847-N35.}}
\affil[1]{Institute of Software Technology, Graz University of Technology.\\ \texttt{apilz@ist.tugraz.at}}

\author[2]{Patrick Schnider}
\affil[2]{Department of Computer Science, ETH Z\"urich.\\ \texttt{patrick.schnider@inf.ethz.ch}}

\maketitle

\begin{abstract}
The centerpoint theorem is a well-known and widely used result in discrete geometry.
It states that for any point set $P$ of $n$ points in $\mathbb{R}^d$, there is a point $c$, not necessarily from $P$, such that each halfspace containing $c$ contains at least $\frac{n}{d+1}$ points of $P$.
Such a point $c$ is called a centerpoint, and it can be viewed as a generalization of a median to higher dimensions.
In other words, a centerpoint can be interpreted as a good representative for the point set $P$.
But what if we allow more than one representative?
For example in one-dimensional data sets, often certain quantiles are chosen as representatives instead of the median.

We present a possible extension of the concept of quantiles to higher dimensions.
The idea is to find a set $Q$ of (few) points such that every halfspace that contains one point of $Q$ contains a large fraction of the points of $P$ and every halfspace that contains more of $Q$ contains an even larger fraction of $P$.
This setting is comparable to the well-studied concepts of weak $\varepsilon$-nets and weak $\varepsilon$-approximations, where it is stronger than the former but weaker than the latter.
We show that for any point set of size $n$ in $\mathbb{R}^d$ and for any positive $\alpha_1,\ldots,\alpha_k$ where $\alpha_1\leq\alpha_2\leq\ldots\leq\alpha_k$ and for every $i,j$ with $i+j\leq k+1$ we have that $(d-1)\alpha_k+\alpha_i+\alpha_j\leq 1$, we can find $Q$ of size $k$ such that each halfspace containing $j$ points of $Q$ contains least $\alpha_j n$ points of $P$.
For two-dimensional point sets we further show that for every $\alpha$ and $\beta$ with $\alpha\leq\beta$ and $\alpha+\beta\leq\frac{2}{3}$ we can find $Q$ with $|Q|=3$ such that each halfplane containing one point of $Q$ contains at least $\alpha n$ of the points of $P$ and each halfplane containing all of $Q$ contains at least $\beta n$ points of $P$.
All these results generalize to the setting where $P$ is any mass distribution.
For the case where $P$ is a point set in $\mathbb{R}^2$ and $|Q|=2$, we provide algorithms to find such points in time $O(n\log^3 n)$.
\end{abstract}

\section{Introduction}
Medians and quantiles are ubiquitous in the statistical analysis and visualization of data.
These notions allow for quantifying how deep some point lies within a one-dimensional data set by measuring how many elements of the data set appear before the point and how many appear after it.
In comparison to the mean, medians and quantiles have the advantage that they only depend on the order of the data points, and not their exact positions, making them robust against outliers.
However, in many applications, data sets are multidimensional, and there is no clear order of the data set.
For this reason, various generalizations of medians to higher dimensions have been introduced and studied.
Many of these generalized medians rely on a notion of depth of a query point within a data set, a median then being a query point with the highest depth among all possible query points.
Several such depth measures have been introduced over time, most famously Tukey depth~\cite{tukey} (also called halfspace depth), simplicial depth, or convex hull peeling depth (see, e.g.,~\cite{aloupis}).
All of these depth measures lead to generalized medians that are invariant under affine transformations.
As for quantiles, only a few generalizations have been introduced (see, e.g., \cite{chaudhuri}).
We propose such a generalization by extending a depth measure to sets with a fixed number of query points and defining a quantile as a set with maximal depth.
The depth measure we extend is Tukey depth: the \emph{Tukey depth} of a point $q$ with respect to a point set $P \subset \R^d$ is the minimal number of points of $P$ in any closed halfspace containing~$q$.
More formally, if $H$ denotes the set of closed halfspaces, then the Tukey depth $\td_P(q)$ of~$q$ with respect to $P$ is
\[
\td_P(q)=\min_{q \in h \in H}\{|h\cap P|\} \enspace .
\]
Similarly, the Tukey depth can also be defined for a mass distribution~$\mu$:
\[
\td_{\mu}(q)=\min_{q \in h \in H}\{\mu(h)\} \enspace .
\]

Here, a \emph{mass distribution} $\mu$ on $\mathbb{R}^d$ is a measure on $\mathbb{R}^d$ such that all open subsets of $\mathbb{R}^d$ are measurable, $0 < \mu(\mathbb{R}^d) < \infty$ and $\mu(S)=0$ for every lower-dimensional subset $S$ of $\mathbb{R}^d$.

The centerpoint theorem states that there is always a point of high depth, i.e., a point $q$ such that for every closed halfspace~$h$ containing $q$ we have $|h\cap P|\geq\frac{|P|}{d+1}$ (or $\mu(h)\geq\frac{\mu(\mathbb{R}^d)}{d+1}$ for masses).
Note that, for $d=1$, such a centerpoint is a median: a median has the property that every halfline containing it contains at least half of the underlying data set.
Quantiles can be interpreted similarly: the $\frac{1}{3}$-quantile and the $\frac{2}{3}$-quantile form a set of two points such that every halfline that contains one of them contains at least $\frac{1}{3}$ of the data set.
Furthermore, a halfline containing both of the points contains at least $\frac{2}{3}$ of the underlying data set.
In particular, halflines containing more points contain more of the data set.
This idea leads to the following generalization of Tukey depth for a set $Q$ of multiple points:
\[
\gtd_P(Q):=\min_{h \in H \colon Q\cap h\neq\emptyset}\left\{\frac{|h\cap P|}{|h\cap Q|}\right\} \enspace .
\]
Again, we can generalize this to mass distributions:
\[
\gtd_{\mu}(Q):=\min_{h \in H \colon Q\cap h\neq\emptyset}\left\{\frac{\mu(h)}{|h\cap Q|}\right\} \enspace .
\]

We prove that there is always a set $Q$ of $k$ points that has generalized Tukey depth $\frac{1}{kd+1}$.
In fact, we prove the following, more general statement:

\begin{restatable}{theorem}{kPointsGeneral}\label{thm:k_points_2d}
{
Let $\mu$ be a mass distribution in $\mathbb{R}^d$ with $\mu(\mathbb{R}^d)=1$.
Let $\alpha_1,\ldots,\alpha_k$ be non-negative real numbers such that $\alpha_1\leq\alpha_2\leq\ldots\leq\alpha_k$ and for every $i,j$ with $i+j\leq k+1$ we have that $(d-1)\alpha_k+\alpha_i+\alpha_j\leq 1$.
Then there are $k$ points $p_1,\ldots, p_k$ in $\mathbb{R}^d$ such that for each closed halfspace~$h$ containing $j$ of the points $p_1,\ldots, p_k$ we have $\mu(h)\geq\alpha_j$.
}
\end{restatable}

Note that, for $d=1$ and $k=2$, the points $p_1$ and $p_2$ correspond to the $\alpha_1$-quantile and the $(1-\alpha_1)$-quantile; for $\alpha_j=\frac{j}{kd+1}$ we get our bound on the generalized Tukey depth, and for $\alpha_1=\ldots=\alpha_k$, the result implies the centerpoint theorem.

Our second result is motivated by interpreting the $\frac{1}{3}$-quantile and the $\frac{2}{3}$-quantile not as two points, but as a one-dimensional simplex.
We then have that every halfline that contains a part of the simplex contains at least $\frac{1}{3}$ of the underlying data set and every halfline that contains the whole simplex contains at least $\frac{2}{3}$ of the underlying data set.
Also for this interpretation we give a generalization to two dimensions:

\begin{restatable}{theorem}{SimplexGeneral}\label{thm:simplex_general}
Let $\mu$ be a mass distribution in $\mathbb{R}^2$ with $\mu(\mathbb{R}^2)=1$.
Let $\alpha$ and $\beta$ be real numbers such that $0 < \alpha\leq\beta$ and $\alpha+\beta=\frac{2}{3}$.
Then there is a triangle~$\Delta$ in $\mathbb{R}^2$ such that
\begin{enumerate}
\item[(1)] for each closed halfplane $h$ containing one of the vertices of $\Delta$ we have $\mu(h)\geq\alpha$ and 
\item[(2)] for each closed halfplane $h$ fully containing $\Delta$ we have $\mu(h)\geq\beta$.
\end{enumerate}
\end{restatable}

Note that this again generalizes centerpoints for $\alpha=\beta$.
However, this result does not give bounds on the generalized Tukey depth of these sets, as, e.g., a halfspace containing two points may still only contain an $\alpha$-fraction of the mass.

Finally, we give algorithms to compute two points satisfying the two-dimensional case of Theorem~\ref{thm:k_points_2d} and three points satisfying Theorem~\ref{thm:simplex_general} in time $O(n\log^3 n)$.

\subparagraph*{Related work}
Another way to view our setting is the following: given a multidimensional data set, we want to find a fixed number of representatives.
The idea of small point sets representing a larger point set has been studied in many different settings.
One of the most famous of those is the concept of $\varepsilon$-nets, introduced by Haussler and Welzl~\cite{haussler_welzl}.
For a range space $(X,R)$, consisting of a set $X$ and a set $R$ of subsets of $X$, an $\varepsilon$-net on $P\subset X$ is a subset $N$ of $P$ with the property that every $r\in R$ with $|r\cap P|\geq \varepsilon |P|$ intersects $N$.
In our setting, where we consider halfspaces, we would choose $X=\mathbb{R}^d$ and $R$ as the set of all halfspaces.
It is known that for this range space, for any point set $P$ there exists an $\varepsilon$-net of size $O(\frac{d}{\varepsilon}\log\frac{d}{\varepsilon})$.
In particular, this bound does not depend on the size of $P$.
Note that we require the $\varepsilon$-net to be a subset of $P$.
If this condition is dropped, we arrive at the concept of \emph{weak} $\varepsilon$-nets.
The fact that the points for the weak $\varepsilon$-net can be chosen anywhere in $\mathbb{R}^d$ allows for very small weak $\varepsilon$-nets for many range spaces.
There has been some work on weak $\varepsilon$-nets of small size.
For halfplanes in $\mathbb{R}^2$ for example, Aronov et al.~\cite{Aronov} have shown that there is always a weak $\frac{1}{2}$-net of two points.
These two points both lie outside of the convex hull of $P$.
They also consider many other range spaces, such as convex sets, disks and rectangles.
Similarly, Babazadeh and Zarrabi-Zadeh \cite{small_weak_3d} construct weak $\frac{1}{2}$-nets of size 3 for halfspaces in $\mathbb{R}^3$.
For two-dimensional convex sets, Mustafa and Ray~\cite{Mustafa} have shown that there is always a weak $\frac{4}{7}$-net of two points; Shabbir~\cite{shabbir} shows how to find two such points in $O(n \log^4 n)$ time.

Another related concept is the concept of $\varepsilon$-approximations:
For a range space $(X,R)$ an $\varepsilon$-approximation on $P\subset X$ is a subset $N$ of $P$ with the property that for every $r\in R$ we have $\left| \frac{|r\cap P|}{|P|}-\frac{|r\cap N|}{|N|}\right|\leq\varepsilon$.
Again, the restriction that $N$ has to be a subset of $P$ can be dropped to get the notion of weak $\varepsilon$-approximations.
Just as for $\varepsilon$-nets, there has been a lot of work on $\varepsilon$-approximations and weak $\varepsilon$-approximations, see \cite{handbook} for a recent survey.
In particular it was shown that for halfspaces in $\mathbb{R}^d$, there always is an $\varepsilon$-approximation of size $O(\frac{1}{\varepsilon^{2-2/(d+1)}})$ \cite{Matousek, MatWelzl}.

While our setting can be considered to be related to weak $\varepsilon$-nets and weak $\varepsilon${-\penalty0\hskip0pt\relax}approximations for range spaces determined by halfspaces, the differences are significant.
As we will discuss here, a halfspace missing all the points of $Q$ may still contain up to half of the points of the initial set, and thus $Q$ qualifies neither as a good weak $\varepsilon$-approximation nor $\varepsilon$-net.

Note that for $|Q|=2$, the condition of Theorem~\ref{thm:k_points_2d} that any halfspace containing all of the points of $Q$ contains at least $\alpha_2 n$ points of $P$ is equivalent to the statement that every halfspace containing more than $(1-\alpha_2)n$ of the points of $P$ contains at least one point of $Q$.
So, $Q$ is a weak $(1-\alpha_2)$-net of $P$.
Furthermore, the condition that any halfspace containing one of the points of $Q$ contains at least $\alpha_1 n$ points of $P$ translates to the statement that every halfspace containing more than $(1-\alpha_1)n$ of the points of $P$ must contain all of $Q$.
Thus, $Q$ is not only a weak $(1-\alpha_2)$-net of $P$ but also has the additional property that all points of $Q$ are somewhat deep within~$P$.
(For two points in the plane, this comes at the cost of having $\varepsilon$ larger than $\frac12$.)
On the other hand, while we require halfspaces containing all points of $Q$ to also contain many points of $P$, we also allow halfspaces containing only one point of $Q$ to contain many points of~$P$.
This separates our concept from weak $\varepsilon$-approximations.
Note that when dealing with halfspaces and $\varepsilon$-nets of size~2, the weak $\frac{1}{2}$-net of Aronov et al.~\cite{Aronov} is actually also a weak $\frac{1}{2}$-approximation.
Similarly, Theorem~\ref{thm:k_points_2d} gives us a weak $(1-\alpha_2)$-approximation of~$P$, with the optimal value being reached when $\alpha_1$ tends to zero (which actually corresponds to the result in~\cite{Aronov}).
Adding more points to $Q$ does not give us a better approximation.
For $d=2$, requiring that for $i+j \leq k+1$ we have $(d-1)\alpha_i + \alpha_j + \alpha_k < 1$ implies $\alpha_1 + 2\alpha_k < 1$, so a halfspace containing no points of $Q$ may contain half of the points of $P$;
we therefore cannot get anything better than a weak $\frac{1}{2}$-approximation.
Also, we do not get anything better than a weak $\frac{1}{2}$-net.

In fact, our setting is very much related to the concept of  one-sided $\varepsilon$-approximants, recently introduced by Bukh and Nivasch~\cite{approximants}:
For a range space $(X,R)$, a \emph{one-sided $\varepsilon$-approximant} on $P\subset X$ is a subset $N$ of $P$ with the property that for every $r\in R$ we have $\frac{|r\cap P|}{|P|}-\frac{|r\cap N|}{|N|}\leq\varepsilon$.
Once again, the restriction that $N$ has to be a subset of $P$ can be dropped to get the notion of weak one-sided $\varepsilon$-approximations.
Note that the only difference to the definition of $\varepsilon$-approximations is that one does not take the absolute value of the difference.
In particular, if $\frac{|r\cap N|}{|N|}>\frac{|r\cap P|}{|P|}$, i.e., if $r$ contains many points of $N$ despite containing only few points of $P$, the difference is negative, and thus smaller than $\varepsilon$.

In their paper, Bukh and Nivasch~\cite{approximants} consider the range space of convex sets.
They show that any point set in $\mathbb{R}^d$ allows a one-sided $\varepsilon$-approximant for convex ranges of size $g(\varepsilon,d)$, where $g(\varepsilon,d)$ only depends on $\varepsilon$ and $d$, but not on the size of $P$.

In a similar reasoning, it makes sense to define an approximation by a set $N$ such that for every $r\in R$ we have $\frac{|r\cap N|}{|N|}-\frac{|r\cap P|}{|P|}\leq\varepsilon$.
Intuitively, if a range $r$ contains a large fraction of the points of $N$, then it is guaranteed to contain a large fraction of the set $P$ we want to approximate.
But here again, our approximation ratio is $\frac{1}{2}$ at best.

\section{Two points}
\label{sec:2points}

We first consider the case where the underlying data is a point set. Motivated by the definition of generalized Tukey depth, we consider $\alpha_1=\frac{1}{5}$ and $\alpha_2=\frac{2}{5}$.
Even though this result is a special case of Theorem \ref{thm:k_points_2d}, we still show its proof for two reasons: first, the Algorithm presented in Section \ref{sec:algorithms} relies heavily on the presented proof and, secondly, the proof already illustrates the main ideas for the proof of Theorem \ref{thm:k_points_2d}.

\begin{theorem}\label{thm:two_points_plane}
Let $P$ be a set of $n$ points in general position in the plane.
Then there are two points $p_1$ and $p_2$ in $\mathbb{R}^2$ such that
\begin{enumerate}
\item[(1)] each closed halfplane containing one of the points $p_1$ and $p_2$ contains at least $\frac{n}{5}$ of the points of $P$ and 
\item[(2)] each closed halfplane containing both $p_1$ and $p_2$ contains at least $\frac{2n}{5}$ of the points of $P$.
\end{enumerate}
\end{theorem}

\begin{proof}
Note that condition (1) is equivalent to the condition that every open halfplane containing more than $\frac{4n}{5}$ of the points of $P$ must contain both $p_1$ and $p_2$.
Similarly, condition~(2) is equivalent to the condition that every open halfplane containing more than $\frac{3n}{5}$ of the points of $P$ must contain one of $p_1$ and $p_2$.
We will now construct two points $p_1$ and $p_2$ satisfying both these conditions.

Let $C$ be the intersection of all open halfplanes containing more than $\frac{4n}{5}$ of the points of~$P$.
Clearly $C$ is convex.
Also, note that $C$ is closed.
The centerpoint region is a strict subset of $C$ and thus $C$ has a non-empty interior.
In order to satisfy condition (1), both $p_1$ and $p_2$ have to be placed in $C$.

Let now $H$ be the set of all open halfplanes containing more than $\frac{3n}{5}$ of the points of $P$.
For any $h_i$ in $H$ let $c_i$ be the intersection of $h_i$ and $C$.
In order to also satisfy condition (2), we need to find two points $p_1$ and $p_2$ such that every $c_i$ contains at least one of them.
To this end, we partition $H$ into two subsets $L$ and $R$.
The set $L$ contains all halfplanes that lie on the left side of their respective boundary lines.
Analogously, $R$ contains all halfplanes that lie on the right side of their respective boundary lines.
For a halfplane $h_i$ that has a horizontal boundary line, we put $h_i$ in $L$ if and only if it lies above its boundary line.

Note that any three halfplanes in $L$ have a non-empty intersection:
Consider the inclusion-minimal halfplane $h \in L$ with horizontal boundary line and its intersection~$r$ with the boundary of the convex hull of~$P$.
As $h$ is open, $r$ is not in $h$.
However, we claim that any point $r'$ in $h$ on the convex hull boundary of $P$ in an $\varepsilon$-neighborhood of $r$ is in any halfplane of $L$.
Indeed, if there was a halfplane in~$L$ not containing $r'$, it would contain a strict subset of the intersection of the convex hull of $P$ with~$h$;
however, this would contradict the minimality of~$h$.
The analogous holds for $R$.

We will now show that for any two halfplanes $h_1$ and $h_2$ in $L$, their corresponding regions $c_1$ and $c_2$ have a non-empty intersection.
The same arguments hold for any two halfplanes in $R$.
Assume for the sake of contradiction that $c_1$ and $c_2$ do not intersect.
As $C$ and $h_1\cap h_2$ are convex, this means that there is an open halfplane $g$ containing more than $\frac{4n}{5}$ of the points of $P$ such that the intersection of the boundary lines of $h_1$ and $h_2$ lies in $\overline{g}$, the complement of $g$ (see \fig{Fig:2intersect}).
In particular, $g\cap h_1$ is a strict subset of $\overline{h_2}$.
As $\overline{g}$ contains strictly fewer than $\frac{n}{5}$ of the points of $P$ and $\overline{h_1}$ contains strictly fewer than $\frac{2n}{5}$ of the points of $P$, $g\cap h_1$ must contain strictly more than $\frac{2n}{5}$ of the points of $P$.
However, being a subset of $\overline{h_2}$, which also contains strictly fewer than $\frac{2n}{5}$ of the points of $P$, this is a contradiction.
Thus, by contradiction, $c_1$ and $c_2$ intersect.

\begin{figure}[t]
\centering
\includegraphics[scale=0.8]{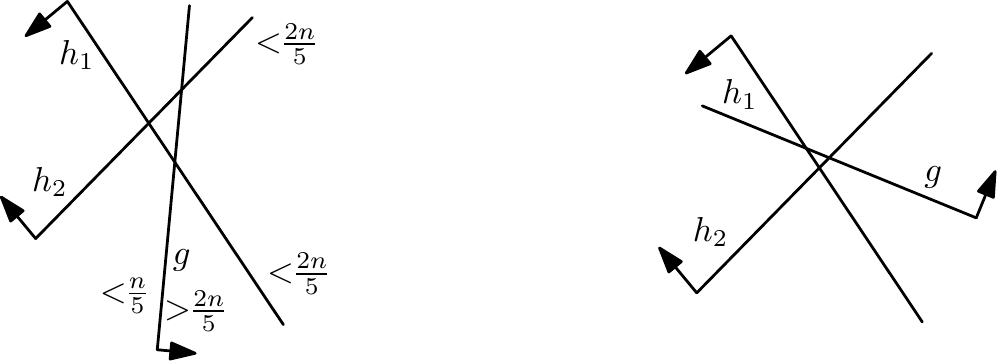}
\caption{Two $c_i$s associated to $L$ must intersect (left).
The intersection is non-empty in other variants (right).}
\label{Fig:2intersect}
\end{figure}

As neither three halfplanes in $L$ nor two halfplanes in $L$ and $C$ have an empty intersection, 
Helly's Theorem entails that there exists a point in both $C$ and all halfplanes in $L$, i.e., all $c_i$s associated to $L$ have a non-empty intersection $D_L$.
Again, the same holds for $R$, with a non-empty intersection $D_R$.
Placing $p_1$ in $D_L$ and $p_2$ in $D_R$, we have thus constructed two points such that the conditions (1) and (2) hold.
\end{proof}

This result is tight in the following sense: There is a point set for which it is not possible to improve both conditions at the same time, that is, it is not possible to find two points such that any halfplane containing one of them contains strictly more than $\frac{n}{5}$ of the points and any halfplane containing both of them contains strictly more than $\frac{2n}{5}$ of the points.
For this consider a set of $n=5k$ point arranged in the following way.
Partition the points into 5 sets $P_1,\ldots, P_5$ of $k$ points each.
Place $P_1,\ldots, P_5$ in such a way that the convex hull of each $P_i$ is disjoint from the convex hull of the union of the other four sets (see \figurename~\ref{Fig:2points_tight}).

\begin{figure}[ht]
\centering
\includegraphics[scale=0.9]{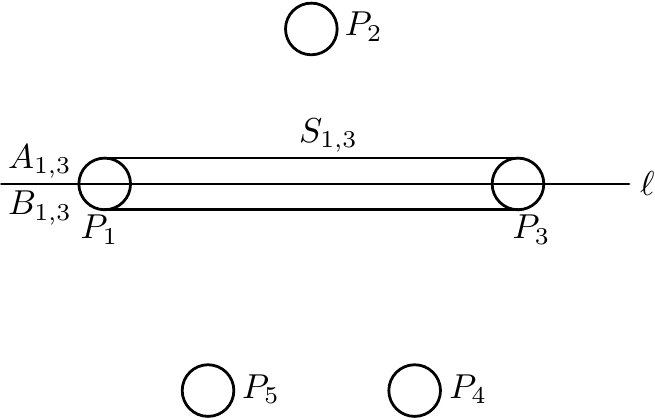}
\caption{A construction for which the bounds of Theorem~\ref{thm:two_points_plane} cannot be improved.}
\label{Fig:2points_tight}
\end{figure}

Denote by $S_{i,j}$ the convex hull $\CH(P_i\cup P_j)$ of $P_i\cup P_j$.
Let $\ell$ be a line through $\CH(P_i)$ and $\CH(P_j)$.
Note that any other set $P_m$ is not separated by $\ell$ (i.e., lies entirely on one side).
Let $A_{i,j}$ be the side of $\ell$ containing fewer of the other sets and let $B_{i,j}$ be the other side.
For any point $q$ in $\CH(P_1\cup\ldots\cup P_5)$ we say that $q$ is \emph{above} $S_{i,j}$ if it is not in $S_{i,j}$ but it is in $A_{i,j}$.
Similarly, for any point $q$ in $\CH(P_1\cup\ldots\cup P_5)$ we say that $q$ is \emph{below} $S_{i,j}$ if it is not in $S_{i,j}$ but it is in $B_{i,j}$.
Suppose, for the sake of contradiction, that there exist two points $p_1$ and $p_2$ such that any halfplane containing one of them contains strictly more than $k$ of the points of $P_1\cup\ldots\cup P_5$ and any halfplane containing both of them contains strictly more than $2k$ of the points of $P_1\cup\ldots\cup P_5$.
Consider two sets $P_i$ and $P_j$ such that $A_{i,j}$ contains exactly one other set.
First we note that neither $p_1$ nor $p_2$ can lie above $S_{i,j}$ as otherwise we can find a halfplane containing that point and only one of the sets, i.e., only $k$ points.
Similarly, we cannot place both $p_1$ and $p_2$ below $S_{i.j}$, as otherwise we can find a halfplane containing both points and only two of the sets, i.e., only $2k$ points.
Also, we must clearly place both $p_1$ and $p_2$ in $\CH(P_1\cup\ldots\cup P_5)$.
Thus, for any two sets $P_i$ and $P_j$ such that $A_{i,j}$ contains exactly one other set, $S_{i,j}$ must contain at least one of $p_1$ and~$p_2$.
However, there are five such $S_{i,j}$ and $P_1,\ldots, P_5$ can be placed in such a way that no three of them have a common intersection.
So no matter how we place $p_1$ and $p_2$, one of the $S_{i,j}$ will be empty.

\section{An arbitrary number of points}
\label{sec:kpoints}

We now strengthen Theorem~\ref{thm:two_points_plane} in four ways: we allow for arbitrarily many query points, we extend it to higher dimensions, we consider mass distributions instead of point sets, and we give a range of possible bounds:

\kPointsGeneral*

We use the following observation, which follows from the fact that for an empty intersection of $d+1$ halfspaces, any point with non-zero mass is in at most $d$ such halfspaces.

\begin{observation}
\label{obs:HellyCondition}
Let $\mu$ be a mass distribution in $\mathbb{R}^d$ with $\mu(\mathbb{R}^d)=1$.
Let $h_1,\ldots, h_{d+1}$ be $d+1$ open halfspaces with $h_1\cap \ldots\cap h_{d+1}=\emptyset$.
Then $\mu(h_1) + \ldots + \mu(h_{d+1}) \leq d$.
\end{observation}

\begin{proof}[Proof of Theorem~\ref{thm:k_points_2d}]
The result is straightforward for $d=1$, so assume $d\geq 2$.
Again the condition that for each closed halfspace $h'$ containing $j$ of the points $p_1,\ldots, p_k$ we have $\mu(h')\geq\alpha_j$ is equivalent to the condition that every open halfspace $h$ with $\mu(h)>1-\alpha_j$ must contain at least $k-j$ of the points $p_1,\ldots, p_k$.
Let $\alpha_0 = 0$.
For $1 \leq j \leq k$, we call an open halfspace $h$ a \emph{$j$-halfspace} if $1-\alpha_{k-j+1}<\mu(h)\leq 1-\alpha_{k-j}$.
Consider the $x_1$-$x_2$-plane, denoted by $X$, and for each vector $v=(v_1,v_2,\ldots,v_d)$ in $\mathbb{R}^d$ let $\pi(v)=(v_1,v_2,0,\ldots,0)$ be the projection of $v$ to $X$.
Let $v_1,\ldots,v_k$ be $k$ unit vectors in $X$ with the property that the angle between any $v_i$ and $v_{i+1}$ is $\frac{2\pi}{k}$.
Note that this implies that also the angle between $v_k$ and $v_1$ is $\frac{2\pi}{k}$.
For each $v_i$ we construct a \emph{principal set} $V_i$ of halfspaces as follows:
For each $j$, consider all $j$-halfspaces.
For any such halfspace~$h$, let $n(h)$ be the normal vector to its bounding hyperplane that points into $h$.
Let $h$ be in $V_i$ if the angle between $\pi(n(h))$ and $v_i$ is at most $\frac{j\pi}{k}$.
If $\pi(n(h))=0$, place $h$ arbitrarily in $j$ of the $V_i$'s.
Note that with this construction each $j$-halfspace is contained in exactly $j$ principal sets.
Thus, if, for each principal set, we can pick a point in all its halfplanes, then each $j$-halfplane contains $j$ points.

It remains to show that the halfspaces in each principal set have a common intersection.
Let $h_1,\ldots,h_{d+1}$ be $d+1$ halfspaces in $V_i$ and assume for the sake of contradiction that they have no common intersection.
Then the positive hull (conical hull) of their projected normal vectors must be $X$, and in particular there are three of them, w.l.o.g.\ $h_1$, $h_2$ and $h_3$, whose projected normal vectors already have $X$ as their positive hull.
Further, among those three halfspaces, there are two of them, w.l.o.g.\ $h_1$ and $h_2$, such that the angles between their projected normal vectors and $v_i$ sum up to more than $\pi$.
If $h_1$ is a $j_1$-halfspace, then by construction of $V_i$ we have that the angle between $\pi(n(h_1))$ and $v_i$ is at most $\frac{j_1\pi}{k}$.
Analogously, if $h_2$ is a $j_2$-halfspace, the angle between $\pi(n(h_2))$ and $v_i$ is at most $\frac{j_2\pi}{k}$.
By the choice of $h_1$ and $h_2$ we thus have $\frac{(j_1+j_2)\pi}{k}>\pi$, which is equivalent to $j_1+j_2>k$, and to $j_1+j_2\geq k+1$,
as $j_1$ and $j_2$ are integers.
By definition of a $j$-halfspace we have
\[
\mu(h_1)+\mu(h_2)>1-\alpha_{k+1-j_1}+1-\alpha_{k+1-j_2}\enspace .
\]
Furthermore we have $\mu(h_i)>1-\alpha_k$ for every $i\in\{1,\ldots,d+1\}$, and thus
\[
\mu(h_1)+\mu(h_2)+\mu(h_3)+\ldots+\mu(h_{d+1})>1-\alpha_{k+1-j_1}+1-\alpha_{k+1-j_2}+(d-1)(1-\alpha_k)\enspace ,
\]
which is equivalent to
\[
(d-1)\alpha_k+\alpha_{k+1-j_1}+\alpha_{k+1-j_2}>d+1-(\mu(h_1) + \ldots + \mu(h_{d+1})) \enspace .
\]
As $k+1-j_1+k+1-j_2=2k+2-(j_1+j_2)\leq k+1$, we have that $(d-1)\alpha_k+\alpha_{k+1-j_1}+\alpha_{k+1-j_2}\leq 1$ and thus $\mu(h_1) + \ldots + \mu(h_{d+1})>d$, which is a contradiction to Observation~\ref{obs:HellyCondition}.
\end{proof}

Setting $\alpha_j=\frac{j}{kd+1}$, we get a bound for the generalized Tukey depth:

\begin{corollary}
Let $\mu$ be a mass distribution in $\mathbb{R}^d$ with $\mu(\mathbb{R}^d)=1$.
Then there exist $k$ points $p_1,\ldots,p_k$ in $\mathbb{R}^d$ %such that $\{p_1,\ldots,p_k\}$ has 
with generalized Tukey depth
$ %\[
\gtd_{\mu}(\{p_1,\ldots,p_k\})=\frac{1}{kd+1} % \enspace .
$. %\]
\end{corollary}

\section{Triangles}
\label{sec:simplex}

As mentioned before, the $\frac{1}{3}$-quantile and the $\frac{2}{3}$-quantile can also be interpreted as a one-dimensional simplex with the property that every halfline that contains a part of the simplex contains at least $\frac{1}{3}$ of the underlying data set and every halfline that contains the whole simplex contains at least $\frac{2}{3}$ of the underlying data set.
For this interpretation, we give a generalization to two dimensions.
For ease of presentation, we only give a proof for point sets instead of mass distributions and for fixed values of $\alpha$ and $\beta$.

\begin{theorem}\label{thm:three_points_plane}
Let $P$ be a set of $n$ points in general position in the plane.
Then there are three points $p_1$, $p_2$ and $p_3$ in $\mathbb{R}^2$ such that
\begin{enumerate}
\item[(1)] each closed halfplane containing one of the points $p_1$, $p_2$ and $p_3$ contains at least $\frac{n}{6}$ of the points of $P$ and 
\item[(2)] each closed halfplane containing all of $p_1$, $p_2$ and $p_3$ contains at least $\frac{n}{2}$ points of~$P$.
\end{enumerate}
\end{theorem}

Note that this can also be interpreted as an instance of Theorem~\ref{thm:k_points_2d} with $\alpha_1=\alpha_2=\frac{1}{6}$ and $\alpha_3=\frac{1}{2}$.
However, as $\alpha_3+\alpha_3+\alpha_1>1$, the precondition of Theorem~\ref{thm:k_points_2d} does not apply.
The proof of this result uses similar ideas as the above proofs.

\begin{proof}[Proof of Theorem~\ref{thm:three_points_plane}]
Let $C$ be the intersection of all open halfplanes containing more than $\frac{5n}{6}$ of the points of $P$.
Just as in the proof of Theorem~\ref{thm:two_points_plane}, condition (1) is equivalent to $p_1$, $p_2$ and $p_3$ lying in $C$.
Similarly, condition (2) is equivalent to the following statement:
for every open halfplane $h$ containing more than $\frac{n}{2}$ of the points of~$P$, $h$ contains at least one of $p_1$, $p_2$ and $p_3$.
For the latter, let $n_i$ be a vector and let $N_i$ be the set of all of these halfplanes that have $n_i$ as their normal vector.
The intersection of all elements of $N_i$ is a closed halfspace $h_i$.
Let $K$ be the set of all of these $h_i$, i.e., for every possible direction of $n_i$.
Then, condition (2) is equivalent to the following statement:
for every halfplane $h$ in $K$, $h$ contains at least one of $p_1$, $p_2$ and $p_3$.
For each such halfplane $h_i$, let $c_i$ be the intersection of $h_i$ and $C$.
Note that $c_i$ is compact.
We thus want to show the claim that we can find three points $p_1$, $p_2$ and $p_3$ such that each $c_i$ contains at least one of them.
Let $H$ be the set of $c_i$s that are minimal with respect to set inclusion.
Clearly, it is enough to show the claim just for the elements of $H$.
Let $N$ be the set of normal vectors of the defining halfplanes of the elements of $H$.
Note that there is a natural mapping from $N$ to a subset of the unit circle $S^1$.
We will say that a subset of $N$ is connected, if its image under this mapping to $S^1$ is connected.

First we show that among any three elements of $H$, two of them intersect.
Let $c_1$, $c_2$ and $c_3$ be elements of $H$, and let $h_1$, $h_2$ and $h_3$ be their associated halfplanes.
Assume for the sake of contradiction that $c_1$, $c_2$ and $c_3$ are pairwise disjoint.
Let $n_i$ be the normal vector of $h_i$.
Let $A$ be the positive hull of $n_1$, $n_2$ and $n_3$.
Then $A$ is either a cone or the whole plane.
See \fig{fig:3points_cases}.
If $A$ is the whole plane, then $h_1$, $h_2$ and $h_3$ have no common intersection.
Otherwise, if $A$ is a cone, then one of $n_1$, $n_2$ and $n_3$ can be described as a positive linear combination of the other two.
In particular, $h_1$, $h_2$ and $h_3$ have a common intersection and one of them is redundant in the description of $h_1\cap h_2\cap h_3$.
We thus consider two cases, namely whether $h_1$, $h_2$ and $h_3$ have a common intersection or not.

\begin{figure}
\centering
\includegraphics[scale=0.8]{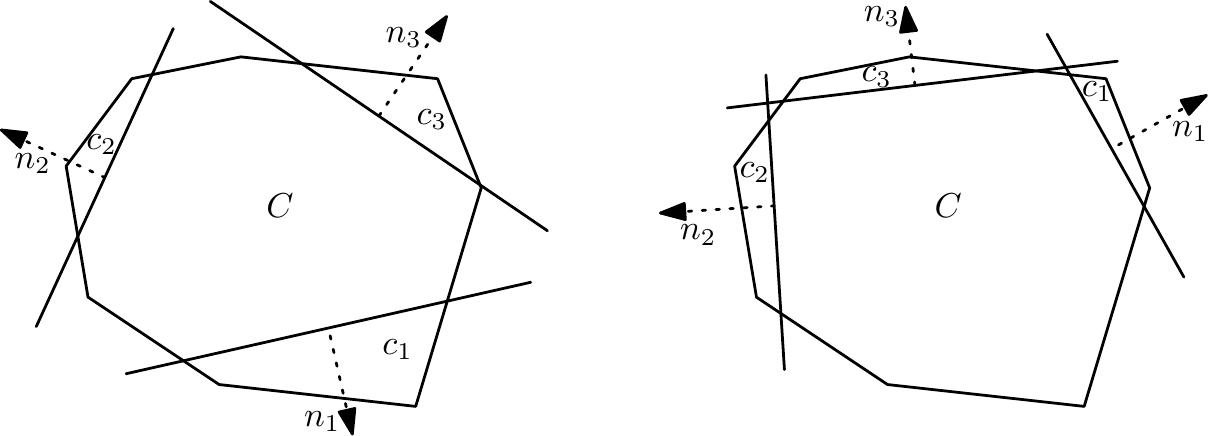}
\caption{Disjoint $c_1$, $c_2$, $c_3$, where the positive hull of the normal vectors spans the whole plane (left) or not (right).}
\label{fig:3points_cases}
\end{figure}

First, assume that $h_1$, $h_2$ and $h_3$ have no common intersection.
Then $h_1$, $h_2$ and $h_3$ partition the plane into seven regions (see \fig{Fig:3points_plane}): $h_i\cap h_j$ for $i\neq j$, $h_i\setminus(h_j\cup h_k)$, for $i,j,k$ all different and $h_1^c\cap h_2^c\cap h_3^c$.
Note that each $h_i\cap h_j$ contains strictly fewer than $\frac{n}{6}$ of the points of $P$, as otherwise the corresponding $c_i$ and $c_j$ intersect.
In particular, $h_2\setminus h_1$ contains more than $\frac{n}{2}-\frac{n}{6}=\frac{n}{3}$ points of $P$.
It follows that $(h_1\cup h_2)^c$ and thus also $h_3\setminus (h_1\cup h_2)$ contains strictly fewer than $\frac{n}{6}$ of the points of $P$.
The number of points in $h_3$ is the sum of the number of points in $h_3\cap h_1$, $h_3\cap h_2$ and $h_3\setminus (h_1\cup h_2)$.
All of these sets contain strictly fewer than $\frac{n}{6}$ of the points of $P$, implying that $h_3$ contains fewer than $\frac{n}{2}$ of the points of $P$, which is a contradiction.
This concludes the first case.

\begin{figure}[ht]
\centering
\includegraphics[scale=0.8]{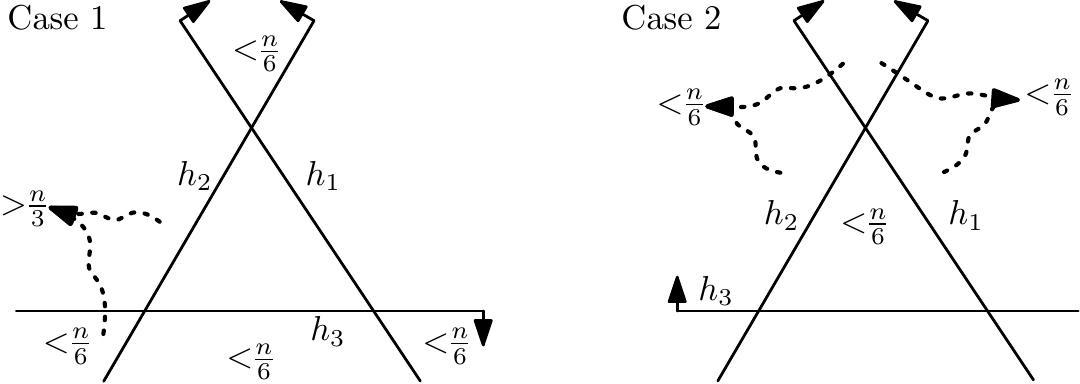}
\caption{No three elements of $H$ can be pairwise disjoint.}
\label{Fig:3points_plane}
\end{figure}

For the second case assume that $h_1$, $h_2$ and $h_3$ have a common intersection and one of the halfplanes is redundant in the description of $h_1\cap h_2\cap h_3$; assume without loss of generality that it is $h_3$.
Just as in the first case, each $h_i\cap h_j$ contains strictly fewer than $\frac{n}{6}$ of the points of $P$.
Again it follows that $h_3\setminus (h_1\cup h_2)$ contains strictly fewer than $\frac{n}{6}$ of the points of $P$.
The sets $h_3\cap h_1$, $h_3\cap h_2$ and $h_3\setminus (h_1\cup h_2)$ cover $h_3$, implying that $h_3$ contains fewer than $\frac{n}{2}$ of the points of $P$, which is again a contradiction.
This concludes the proof that among any three elements of $H$, two intersect.

It remains to show that we can find three points $p_1$, $p_2$ and $p_3$ such that each element of $H$ contains at least one of them.
This can be achieved by picking one element of $H$ and placing two points $p_1$ and $p_2$ at the extreme intersection points with the boundary of $C$;
since any three elements of $H$ intersect, any two elements not containing $p_1$ and $p_2$ must intersect and we may apply Helly's theorem in dimension one.
However, we actually have more flexibility in choosing $p_1$.%, which will be useful for our algorithm.
Note that the normal vectors pointing into the halfplanes defining the elements of $H$ define a circular order on $H$.
Place $p_1$ at a topmost point of the boundary of $C$.
Let $h_1$ be the first element of $H$ in counterclockwise direction whose defining halfplane does not contain $p_1$ in its interior.
Place $p_2$ at the intersection of the defining line of $h_1$ with the boundary of $C$ that is furthest in counterclockwise direction from $p_1$.
Since $h_1$ is minimal, any element of $H$ intersecting $h_1$ contains either $p_1$ or $p_2$.
(Note that so far $h_1$ does not contain any of $p_1$ or $p_2$.)
Therefore, all elements of $H$ that do not intersect $h_1$ have a common intersection, in which we place $p_3$.
Recall that the defining halfplanes are open and therefore there is no element of $H$ intersecting $h_1$ in a single point.
We therefore may move $p_2$ slightly in clockwise direction, such that it is also contained in $h_1$.
\end{proof}

The general statement can be proved analogously:
\SimplexGeneral*

\section{Construction in the plane}
\label{sec:algorithms}
In this section, we describe algorithms for constructing the points described in Theorems~\ref{thm:two_points_plane} and~\ref{thm:three_points_plane}.
We first observe that the convex regions defined by the intersections of the half-planes in sets like $L$ and~$R$ in the proof of Theorem~\ref{thm:two_points_plane} correspond to levels in the dual line arrangement.
We use the duality $p^* = (y=kx+d) \iff p = (k, d)$ that maps a point $p$ to a line $p^*$.
The \emph{$k$-level} of a line arrangement is the set of points with exactly $k-1$ lines below it and not more than $n-k$ lines above it.
(It thus consists of segments of the line arrangement.)
Suppose we are given $\alpha_1$ and $\alpha_2$, s.t.\ $0 < \alpha_1 \leq \alpha_2$ and $\alpha_1 + 2\alpha_2 = 1$.
Let $U$ be the set of open halfplanes that are above their boundary lines and contain more than $(1-\alpha_2)n$ points of $P$, and let $D_U$ be their intersection.
A point $p$ is in $D_U$ if there is no line through it having at least $\floor{(1-\alpha_2)n+1}$ points of $P$ above it.
If the dual line $p^*$ of $p$ contains a point $\ell^*$ below the $\ceil{\alpha_2 n}$-level of the dual line arrangement of $P$, then $p$ has a supporting line~$\ell$ with more than $(1-\alpha_2)n$ points of $P$ above it.
Since a line has a point below that level if and only if it intersects the interior of its convex hull, the interior of the convex hull of the $\ceil{\alpha_2 n}$-level thus excludes exactly those lines whose primal points are not in~$D_U$.
The supporting lines of the segments of the convex hull of the $\ceil{\alpha_2 n}$-level give the primal points that bound~$D_U$.
Matou\v{s}ek~\cite{matousek_center} describes an algorithm for constructing the $k$-level of a line arrangement in $O(n \log^4 n)$ time.
The \emph{$k$-hull} of a set $P$ of $n$ points in the plane is the set of points $p$ in $\mathbb{R}^2$ such that any closed halfplane defined by a line through $p$ contains at least $k$ points of~$P$.
The set~$C$ in the proof of Theorem~\ref{thm:two_points_plane} is the intersection of all open halfplanes containing more than $\frac{4n}{5}$ points.
$C$ is thus the $\ceil{\frac{n}{5}}$-hull of~$P$.
The $k$-hull of $P$ is obtained by computing the convex hulls of the $k$-level and the $(n-k)$-level of the dual line arrangement of~$P$, which give the upper and lower envelope of the $k$-hull~\cite{matousek_center}.
To construct the points from Theorems~\ref{thm:two_points_plane} and~\ref{thm:three_points_plane} (without explicitly constructing the levels), we use Matou\v{s}ek's algorithmic tools from~\cite{matousek_center}.

\begin{lemma}[{Matou\v{s}ek~\cite[Lemma~3.2]{matousek_center}}]\label{lem:matousek_tangent}
In an arrangement of $n$ lines, let $\gamma$ be the boundary of the convex 
hull of the lines on or below the $k$-level.
Given the arrangement, $k$, and a point~$p$, one can find the tangent to $\gamma$ passing through $p$ and touching $\gamma$ to the right of $p$ (if it exists) in time $O(n \log^2 n)$.
\end{lemma}

\begin{lemma}\label{lem:bitangent}
Given an arrangement of $n$ lines and two numbers $k < l \leq n$, as well as a halfplane $h$, a line separating
the $k$-level from the intersection of $h$ with the $l$-level can be found in $O(n \log^3 n)$ time, if it exists.
The separating line is tangent to both level parts and, from left to right, first intersects the $k$-level and then the relevant part of the $l$-level.
\end{lemma}
\begin{proof}
Let $\gamma$ be the boundary of the convex hull of all points below the $k$-level, and let $\nu$ be the intersection of $h$ with the $l$-level.
Note that $\nu$ might not be connected.
Suppose we want our line to be the counterclockwise bitangent of $\gamma$ and $\nu$ (i.e., from left to right, it first intersects $\gamma$, which has no point above it, and then $\nu$).
Our algorithm works by obtaining tangents to $\nu$ through points on $\gamma$.
Matou\v{s}ek's $O(n \log^2 n)$ algorithm for determining the tangent to a level through a given point that is to the right of that point~\cite[Lemma~3.2]{matousek_center} (our Lemma~\ref{lem:matousek_tangent}) also directly works for parts of a level such as~$\nu$: It requires a sub-algorithm that decides in $O(n \log n)$ time whether a given line~$\ell$ intersects the level (or, in our case, the partial level $\nu$).
This can be done by sorting the intersection of the lines of the arrangements along $\ell$ (see also~\cite[Lemma~3.1]{matousek_center}) as well as along the line bounding~$h$;
$\ell$ either intersects the relevant part of $\nu$, or we can compare the intersection of $h$ with $\ell$ to the intersections of $h$ with $\nu$ to determine whether there is a point of $\nu$ below $\ell$.

Suppose first we are given $\gamma$.
(It requires $O(n \log^4 n)$ time though to obtain it, so we eventually get rid of this assumption.)
The convex hull of a level is known to have at most $n$ vertices~\cite[Lemma~2.1]{matousek_center}.
For a point $p$ on $\gamma$, we can find in $O(n \log^2 n)$ time the point $q$ on $\nu$ such that the line $pq$ has no point on $\nu$ below it.
We can thus find, by binary search on the $O(n)$ vertices of~$\gamma$, a vertex $p$ with $q$ on $\nu$ such that $pq$ separates $\gamma$ and $\nu$.
This gives an $O(n \log^4 n)$ time algorithm for obtaining the bitangent.
To improve on that bound, we need to get rid of the explicit construction of $\gamma$ to find the tangents to~$\nu$.

To this end, we consider Matou\v{s}ek's algorithm for constructing the convex hull boundary~$\gamma$ of a level and compute only the relevant part (see~\cite[Section~4]{matousek_center}).
In particular, the algorithm works by finding, for a constant $c$ and two vertical lines, $(c-1)$ further vertical lines between the given ones such that there are at most $n^2/c$ crossings of the arrangement between two of these verticals.
This can be done in $O(n)$ time (as described in~\cite{construction_of_epsilon_nets}).
The tangents on~$\gamma$ at the intersection points with the vertical lines can be computed in $O(n \log^3 n)$ time~\cite[Lemma~3.3]{matousek_center}.
It is shown in~\cite{matousek_center} that, when choosing $c=64$, there are at most $n/2$ lines of the arrangement relevant for the construction of $\gamma$ between two such vertical lines, and these lines can be found in $O(n)$ time.
The original algorithm proceeds recursively within each interval defined by two neighboring vertical lines after removing the non-relevant lines.
In our adaption, however, we find the interval that contains the point $p$ on $\gamma$ such that a tangent to $\gamma$ through the vertex $p$ with $q$ on $\nu$ such that $pq$ separates $\gamma$ and~$\nu$.
(We do this by considering the tangent to $\gamma$ at each of the constant number of intersection of a vertical line with $\gamma$.)
When we have found this interval, we can prune $n/2$ of the lines and recurse inside this interval.
Note, however, that we cannot prune the set of lines when looking for a tangent to $\nu$.
Thus, in each recursive call, we need $O(n \log^2 n)$ time for computing the tangent.
As the recursion depth is $O(\log n)$, this amounts to $O(n \log^3 n)$ in total.
Also, for $n_i$ lines during the $i$th recursion, we need $O(n_i \log^3 n_i) \subseteq O(n_i \log^3 n)$ time for determining the intervals.
As $n_i$ decreases geometrically, this also amounts to $O(n \log^3 n)$.
This is the total running time for finding the bitangent, as claimed.
\end{proof}

We call such a line the \emph{counterclockwise bitangent} of the two subsets of the plane (i.e., the intersection with the region not above it has smaller $x$-coordinate than the intersection with the region not below it).
Note that by mirroring the plane horizontally or vertically, the lemma also provides other types of bitangents.
\fig{fig:levels} shows an example.

\begin{figure}
\centering
\includegraphics[page=6]{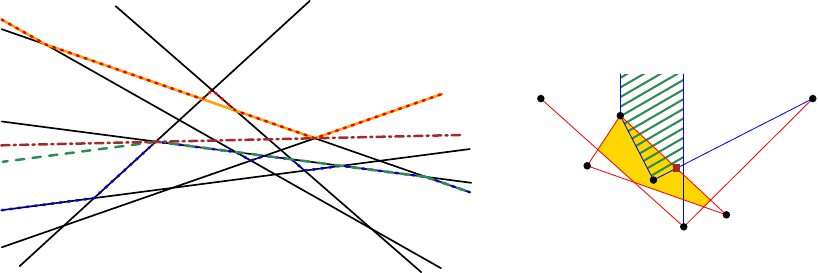}
\caption{A counterclockwise bitangent (brown, dash-dotted) between the $\ceil{\frac{2n}{5}}$-level (blue) and the $\floor{\frac{4n}{5}}$-level (red) of an arrangement of seven lines (left).
The primal point configuration is shown to the right; there, the orange region corresponds to the $\ceil{\frac{n}{5}}$-hull 
$C$, and the hatched green region corresponds to~$D_U$.
Observe that there can be vertices of $D_U$ outside of~$C$.}
\label{fig:levels}
\end{figure}

Lemma~\ref{lem:bitangent} can also be obtained using the framework of Langerman and Steiger~\cite{langerman_steiger}, similar to the computation of the Oja depth in~\cite{oja}.
(We merely sketch its application in this paragraph.)
Let $P$ be the primal point set of the arrangement.
In the formulation of~\cite{oja}, if there exists a function $f : \mathbb{R}^2 \to \mathbb{R}$ that has the minimum on the intersection of supporting lines of point pairs of $P$ with the property that, for a given point $a$, we can find in $T(n)$ time a \emph{witness halfplane} through $a$ such that for all points $q$ in that halfplane, we have $f(q) \geq f(a)$, then there is an $O(T(n) \log^2 n + n \log^3 n)$ time algorithm for finding such a minimum point.
We can test in $O(n \log n)$ time whether the dual $a^*$ of $a$ intersects $\gamma$ or $\nu$, defining a witness halfplane using the primal of such an intersection point (as described at the beginning of the previous proof), or the intersection point of $a^*$ with the line defining~$h$.
If the dual line $a^*$ separates $\gamma$ from $\nu$, we need a separator with a larger slope to obtain a counterclockwise bitangent, and thus define a witness halfplane to the right of a vertical line through~$a$.

Lemma~\ref{lem:bitangent} can now be used to obtain the following result.

\begin{theorem}\label{thm:two_points_algorithm}
Given a set $P$ of $n$ points in the plane, two points satisfying the conditions of Theorem~\ref{thm:two_points_plane} can be constructed in time $O(n \log^3 n)$.
\end{theorem}
\begin{proof}
To find a point $p_1$ in the intersection of $C$ and $D_U$, observe first that we can restrict our attention in the dual to the convex hull of the points above the $\floor{(1-\alpha_1)n}$-level of the dual line arrangement.
This is because any primal line with more than $(1-\alpha_1)n$ points above it (which corresponds to a dual point below the $\ceil{\alpha_1 n}$-level) also defines a halfplane in $U$.
A point in the intersection of $D_U$ and $C$ thus corresponds to a line on or above the $\ceil{\alpha_2 n}$-level and on or below the $\floor{(1-\alpha_1)n}$-level.
We find a bitangent to these two levels in $O(n \log^3 n)$ time using Lemma~\ref{lem:bitangent} (with $h = \mathbb{R}^2$).
The primal point of this line is~$p_1$; see the point indicated by the brown box in \fig{fig:levels}~(right).
We obtain~$p_2$ analogously.
\end{proof}

\begin{theorem}\label{thm:three_points_algorithm}
Three points as described in Theorem~\ref{thm:three_points_plane} can be computed in time $O(n \log^3 n)$.
\end{theorem}
\begin{proof}
Consider the dual line arrangement of the point set.
The points $p_1, p_2, p_3$ dualize to three lines $p_1^*, p_2^*, p_3^*$ that are between the $\ceil{\frac{n}{6}}$-level and the $\floor{\frac{5n}{6}}$-level of the arrangement s.t.\ every point on the middle level has at least one of these lines above it and one of these lines below it.
(We assume for simplicity that $n$ is odd and the \emph{middle level} is the $\floor{\frac{n}{2}}$-level of the arrangement;
if $n$ is even, one has to consider the points between the $\frac{n}{2}$-level and the $(\frac{n}{2}+1)$-level.)
Theorem~\ref{thm:three_points_plane} asserts that such lines exist, and its proof tells us that we can choose one of these lines to be an arbitrary tangent of one of the levels not intersecting the interior of the other one.
We denote by $\gamma_b$ and $\gamma_t$ the convex hull boundaries of the points on or below the $\ceil{\frac{n}{6}}$-level and of the points on or above the $\floor{\frac{5n}{6}}$-level, respectively.
%Let $p^*_1$, $p^*_2$, and $p^*_3$ be the three lines we have to place.

We let $p^*_1$ be the clockwise bitangent of $\gamma_b$ and $\gamma_t$, which we can obtain in $O(n \log^3 n)$ time using Lemma~\ref{lem:bitangent}.
For simplicity of explanation, we also compute the counterclockwise bitangent~$\ell$.
(This step may be omitted in an actual implementation, but assuming it to be given facilitates the explanation and does not change the asymptotic running time.)

\begin{figure}
\centering
\includegraphics[page=6]{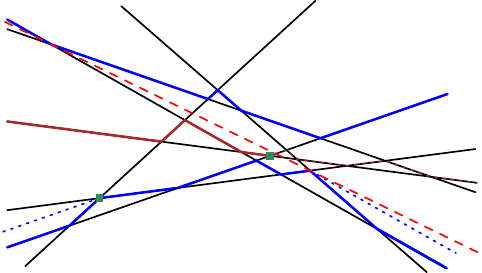}
\caption{An arrangement of seven lines with the $\ceil{\frac{n}{6}}$-level and $\floor{\frac{5n}{6}}$-level (blue) and the clockwise bitangent $p^*_1$ (red dashed) between them.
The green boxes indicate the two points defining the counterclockwise bitangent between the $\ceil{\frac{n}{6}}$-level and $\mu_1$ (brown).
}
\label{fig:bitangents}
\end{figure}

The line $p^*_1$ intersects the middle level of the arrangement.
Let $\mu_1$ be the parts of the middle level below $p^*_1$, and $\mu_2$ be the part above it.
Note that each of these parts may be disconnected.
Using Lemma~\ref{lem:bitangent}, we search for the counterclockwise bitangent between $\gamma_b$ (or, equivalently, the $\ceil{\frac{n}{6}}$-level) and $\mu_1$ (which is the intersection of the middle level with a halfspace defined by $p^*_1$) in $O(n \log^3 n)$ time.
If it exists, and its intersection point with $\gamma_b$ is between the intersections of $\gamma_b$ with $p^*_1$ and~$\ell$, we choose this line to be $p^*_2$.
Otherwise, we continue our search on $\gamma_t)$ in the same way (i.e., we look for the counterclockwise bitangent between $\gamma_t$ and $\mu_1$).
The line $p^*_3$ can be found in an analogous manner.
\end{proof}

\section{Conclusion}
We proposed a generalization of quantiles in higher dimensions based on a generalization of Tukey depth to multiple points.
Our bounds and algorithms seem merely being a first step in this direction and we can identify several interesting open problems.
Except for special cases of Theorem~\ref{thm:k_points_2d}, we do not believe that our bounds are tight and particularly expect significantly better bounds in higher dimensions.
Naturally, there are many other range spaces for which this problem could be considered, e.g., convex sets, like in~\cite{approximants}.

From an algorithmic point of view, the bottleneck for the running time of our approach is Lemma~\ref{lem:bitangent}.
The current methods result in $O(n \log^3 n)$ time.
While solutions to such kinds of problems can usually only be verified in $\Theta(n \log n)$ time (see, e.g.,~\cite{aloupis_lower,steiger}), 
a linear-time algorithm, like for centerpoints~\cite{jadhav}, is conceivable.
For arbitrarily many points, it seems tedious but doable to apply similar approaches as in the proof of Theorem~\ref{thm:two_points_algorithm}.
Is there a good bound on the running time independent of the size of $|Q|$?

\subparagraph{Acknowledgments.}
We thank Emo Welzl for initiating discussions on this topic, as well as anonymous reviewers for helpful comments.

\bibliographystyle{plainurl} %
\bibliography{refs}

\end{document}